\documentclass[runningheads]{llncs}

\usepackage[toc,page]{appendix}

\usepackage{amsmath}
\usepackage{latexsym}
\usepackage{amssymb}
\usepackage{verbatim}
\usepackage{setspace}
\usepackage{qtree}
\usepackage{enumitem}
\usepackage{xcolor}
\usepackage{url}
\usepackage{stmaryrd}

\usepackage{authblk}

\usepackage[margin=1.4in]{geometry}

\DeclareMathOperator{\subf}{Sf}

\DeclareMathOperator{\dom}{\mathrm{dom}}

\DeclareMathOperator{\tp}{\mathrm{tp}}

\DeclareMathOperator{\fo}{\mathcal{FO}}
\DeclareMathOperator{\gf}{\mathcal{GF}}
\DeclareMathOperator{\ugf}{\mathcal{UGF}}

\DeclareMathOperator{\uf1}{\mathcal{UF}_1}
\DeclareMathOperator{\guf1}{\mathcal{UGF}_1}
\DeclareMathOperator{\bisim}{\mathcal{Z}}
\DeclareMathOperator{\modelA}{\mathfrak{A}}
\DeclareMathOperator{\modelB}{\mathfrak{B}}
\DeclareMathOperator{\amalgam}{\mathfrak{U}}

\begin{document}

\title{Uniform Guarded Fragments}
%
%
\author{Reijo Jaakkola}
\authorrunning{R. Jaakkola}
%
\institute{Tampere University, Finland\\
\email{reijo.jaakkola@tuni.fi}\\
\url{https://reijojaakkola.github.io/}}
\maketitle              
\begin{abstract}
In this paper we prove that the uniform one-dimensional guarded fragment, which is a natural polyadic generalization of the guarded two-variable logic, has the Craig interpolation property. We will also prove that the satisfiability problem of uniform guarded fragment is \textsc{NExpTime}-complete.


\end{abstract}

\section{Introduction}

The guarded fragment $\gf$ is a well studied fragment of first-order logic $\fo$, which was introduced by Andréka, van Benthem and Németi \cite{Andrka1998ModalLA} as a generalization of modal logic. Informally speaking, $\gf$ is obtained from $\fo$ by requiring that all quantification must be relativised by $\fo$-atoms, which is motivated by the observation that "quantificaction" in modal logics is relativised by accessability relations. Like modal logic, $\gf$ behaves well both computationally and model-theoretically. In particular, it is decidable, it has a (generalized) tree-model property and it satisfies various preservation theorems \cite{Andrka1998ModalLA,gradel99}.

We say that a logic $\mathcal{L}$ has Craig interpolation property (CIP), if for every two formulas $\varphi$ and $\psi$ of $\mathcal{L}$ we have that if $\varphi \models \psi$, then there exists a third formula --- the interpolant --- $\chi$ of $\mathcal{L}$, so that $\varphi\models \chi$, $\chi \models \psi$ and $\chi$ contains only relation symbols which occur in both $\varphi$ and $\chi$. CIP is widely regarded as a property that a "nice" logic should have and for (reasonable logics with compactness) it implies several other desirable model-theoretic properties such as Projective Beth Definability and Robinson's consistency theorem \cite{Andrka1998ModalLA,10.2307/40271083,Jung2021LivingWB,InterpolationInModalLogic}.

It is well-known that various modal logics have CIP \cite{Andrka1998ModalLA,GABBAY,InterpolationInModalLogic}, while $\gf$ fails to have it \cite{Hoogland2002InterpolationAD}. This is somewhat surprising, given that $\gf$ is a very natural generalisation of modal logic, and certainly raises the question of how the syntax of $\gf$ should be modified so as to obtain a logic which does have CIP, and which also behaves well both computationally and model-theoretically. One option would be to extend further the expressive power of $\gf$, and in this direction we have the \emph{guarded negation fragment}, which has CIP, is decidable and shares with $\gf$ various desirable model-theoretic properties \cite{Brny2018SOMEMT}.

The other option (and the one which is more relevant for this paper) is to investigate fragments of $\gf$. In this direction we also have a positive result, namely that $\gf^2$ --- the two-variable fragment of $\gf$ --- has CIP \cite{Hoogland2002InterpolationAD}. Given this result, it is natural to ask whether there exists a polyadic extension of $\gf^2$ which would also have CIP, where by a polyadic extension we mean intuitively a logic which contains $\gf^2$ and can express non-trivial properties of polyadic relations. Indeed, it seems rather unlikely that there would not be such an extension, since it is well-known that there are polyadic modal logics which have CIP \cite{Andrka1998ModalLA}.

In \cite{OneDimensionalFragment} the authors introduced the \emph{uniform one-dimensional fragment} $\uf1$, which is a very natural polyadic extension of the two-variable fragment $\fo^2$ of $\fo$. Roughly speaking, $\uf1$ is obtained from $\fo$ by requiring that each maximal existential (or universal) block of quantifiers leaves at most one variable free and that when forming boolean combinations of formulas with more than one free variable, the formulas need to have exactly the same set of variables. Formulas satisfying the first restriction are called \emph{one-dimensional}, while formulas satisfying the second restriction are called \emph{uniform}. In \cite{Kieronski2014ComplexityAE} it was proved that $\uf1$ has the finite model property and the complexity of its satisfiability problem is \textsc{NExpTime}-complete, which is the same as for $\fo^2$ \cite{GKV97}. The research around $\uf1$ and its variants has been quite active, see for instance \cite{ORDEREDFRAGMENTS,F1WORDS,Kieronski2019OnedimensionalGF,UF1ONEEQUIVALENCE,F1TREES}.

Given that $\uf1$ is a polyadic extension of $\fo^2$, the guarded $\uf1$ is a natural candidate for being a polyadic extension of $\gf^2$ with CIP. As the first main result of this paper we will prove that guarded $\uf1$ does, in fact, have CIP. Our proof follows closely the argument given in \cite{Hoogland2002InterpolationAD} for proving that $\gf^2$ has CIP, the main technical difference being that the proof presented in \cite{Hoogland2002InterpolationAD} uses crucially the fact that in the case of $\gf^2$ we can assume live sets to have size at most two, while in our case we have to deal with live sets of arbitrary size.

Since the research around modal-like fragments of $\fo$ is largely motivated by the fact that their satisfiability problems are often decidable, it is natural to also study the complexity of the satisfiability problem of the guarded $\uf1$, which was in fact already done in \cite{Kieronski2019OnedimensionalGF}. More precisely it was proved in \cite{Kieronski2019OnedimensionalGF} that the satisfiability problem of one-dimensional $\gf$ is in \textsc{NExpTime}, while it is already \textsc{NExpTime}-hard for guarded $\uf1$. These results left open the problem of determining the complexity of uniform $\gf$ and as the second main result of this paper we will prove that the satisfiability problem of uniform $\gf$ is also in \textsc{NExpTime} (and hence it is \textsc{NExpTime}-complete).

The structure of this paper is as follows. After the preliminaries in section \ref{PreliminariesSection}, we define a notion of bisimulation for $\guf1$ and establish its basic properties in Section \ref{BisimulationSection}. After this we will prove that $\guf1$ has CIP in Section \ref{CipProofSection}. In Section \ref{ComplexityUniformGfSection} we will establish that the complexity of the satisfiability problem of uniform $\gf$ is \textsc{NExpTime}-complete. The final Section will list some new problems that the research conducted in this paper raises.

\section{Preliminaries}\label{PreliminariesSection}

\subsection{Notation}

In this paper we will work with vocabularies which do not contain constants and function symbols. We will also assume that there are no relation symbols of arity $0$. We will use the Fraktul capital letters to denote structures, and the corresponding Roman letters to denote their domains. Given a model $\modelA$ and $C\subseteq A$, we will use $\modelA \upharpoonright C$ to denote the restriction of $\modelA$ to the set $C$. Given two structures $\modelA$ and $\modelB$, we will use $\modelA \leq \modelB$ to denote that $\modelA$ is a substructure of $\modelB$.

Occasionally we will identify tuples $\overline{a} = (a_1,\dots,a_n)$ with sets $\{a_1,\dots,a_n\}$, which allows us to use notations such as $b\in \overline{a}$ and $\overline{a} = X$, where $X$ is a set. Given two tuples $\overline{a}$ and $\overline{b}$ of the same length, we will use $\overline{a} \mapsto \overline{b}$ and $p:\overline{a} \to \overline{b}$ to denote the mapping induced by the relation $a_i \mapsto b_i$. Given a tuple $\overline{a} = (a_1,\dots,a_n)$ and a unary function $f$, we will use $f(\overline{a})$ to denote the tuple $(f(a_1),\dots,f(a_n))$. Given a positive integer $n$ we will denote $[n] = \{1,\dots,n\}$. Finally, if $\overline{a} = (a_1,\dots,a_n)$ and $k\geq n$ and $\mu:[k] \to [n]$ is a surjection, we will use $\overline{a}_\mu$ to denote the tuple $(a_{\mu(1)},\dots,a_{\mu(k)})$.

\subsection{Types and tables}

In the following definitions we will follow fairly closely \cite{Kieronski2014ComplexityAE}. Let $\sigma$ be a vocabulary. Given a set $X = \{x_1,\dots,x_n\}$ of distinct variables and a $k$-ary relation $R\in \sigma$, we say that an atomic formula $R(x_{i_1},\dots,x_{i_k})$ is an $X$-\emph{atom} over $\sigma$, if $X = \{x_{i_1},\dots,x_{i_k}\}$. If $\alpha$ is an $X$-atom, then $\alpha$ and $\neg \alpha$ are both $X$-\emph{literals} over $\sigma$. A $1$-\emph{type} over $\sigma$ is a maximal satisfiable set of $\{x\}$-literals over $\sigma$. We identify $1$-types $\pi$ with conjunctions of their elements
\[\bigwedge \pi\]
A $k$-\emph{table} is a tuple $\langle \rho, \pi_1,\dots,\pi_k\rangle$, where each $\pi_\ell$ is a $1$-type over $\sigma$, while $\rho$ is a maximal satisfiable set of $\{x_1,\dots,x_k\}$-literals over $\sigma$. We identify $k$-tables $\langle \rho, \pi_1, \dots, \pi_k\rangle$ and conjunctions 
\[\bigwedge \rho \land \bigwedge_{1\leq \ell\leq k} \pi_\ell.\]

Let $\modelA$ be a $\sigma$-model. Given a $1$-type $\pi$ over $\sigma$, we say that $a\in A$ realizes $\pi$ if $\pi$ is the unique $1$-type so that $\modelA \models \pi[a]$; we denote by $\tp_{\modelA}^\sigma[a]$ the (unique) $1$-type $\pi$ over $\sigma$ which is realized by $a$ in $\modelA$. For \emph{distinct} elements $a_1,\dots,a_k \in A$ we will use $\tp_{\modelA}^\sigma[a_1,\dots,a_k]$ to denote the (unique) $k$-table over $\sigma$ which is realized by the tuple $(a_1,\dots,a_k)$.

\subsection{Syntax of uniform fragments of $\gf$}

Given a vocabulary $\sigma$, we define $\mathcal{GF}[\sigma]$ to be the smallest set $\mathcal{F}$ which satisfies the following requirements.

\begin{itemize}
    \item $\mathcal{F}$ contains all the atomic formulas over $\sigma$, which includes also equalities between variables.
    \item If $\varphi, \psi \in \mathcal{F}$, then $\neg \varphi \in \mathcal{F}$ and $(\varphi \land \psi) \in \mathcal{F}$.
    \item If $\psi(\overline{x}) \in \mathcal{F}$, where each free variable of $\psi$ occurs in the tuple $\overline{x}$, then 
    \[\exists \overline{y} (\alpha(\overline{x}) \land \psi(\overline{x}))\in \mathcal{F},\]
    where $\overline{y}\subseteq \overline{x}$ and $\alpha$ is an atomic formula over $\sigma$.
\end{itemize}

\noindent If the vocabulary $\sigma$ is irrelevant or known from the context, then we will simply use $\gf$ to denote $\gf[\sigma]$.

Next we will give a formal definitions for the syntactical notions of one-dimensionality and uniformity. We will start by making the technical remark that we will define recursively the set of subformulas $\subf(\varphi)$ of $\varphi \in \gf$ in a standard way, expect that for formulas of the form $\varphi := \exists \overline{y} (\alpha(\overline{x}) \land \psi(\overline{x}))$, we define $\subf(\varphi)$ to be
\[\{\exists \overline{y} (\alpha(\overline{x}) \land \psi(\overline{x}))\} \cup \subf((\alpha(\overline{x}) \land \psi(\overline{x}))).\]

\begin{definition}
    Let $\varphi \in \gf$ be a formula. We say that $\varphi$ is one-dimensional, if every subformula of $\varphi$ of the form
    \[\exists \overline{y} (\alpha (\overline{x}) \land \psi(\overline{x}))\]
    has at most one free variable. In other words each maximal sequence of (guarded) existential quantification leaves at most one variable free.
\end{definition}

Next we will define what it means for a formula of $\gf$ to be uniform. The precise definition turns out to be somewhat technical, and we will start with the following auxiliary definition.

\begin{definition}
    Let $X$ be a (possibly empty) set of variables and let $\sigma$ be a vocabulary. A relative $X$-atom over $\sigma$ is a formula $\psi$ of $\gf[\sigma]$ which satisfies one of the following conditions.
    \begin{enumerate}
        \item $\psi$ is a sentence.
        \item $\psi$ has a one free variable which belongs to $X$.
        \item $\psi$ is of the form $x = y$, where $x,y\in X$.
        \item $\psi$ is an $X$-atom over $\sigma$.
        \item $\psi$ is of the form $\exists \overline{z} (\alpha(\overline{x}) \land \psi(\overline{x}))$ and the set of its free variables is precisely $X$.
    \end{enumerate}
\end{definition}

With the aid of this definition we are in a position where we can define the notion of uniformity formally.

\begin{definition}
    Let $\varphi \in \gf[\sigma]$ be a formula. We say that $\varphi$ is uniform, if every subformula $\psi$ of $\varphi$ is a boolean combination of relative $X$-atoms, where $X$ is the set of free variables of $\psi$.
\end{definition}

\begin{remark}
    Consider a uniform quantifier-free formula $\psi(x_1,\dots,x_k)$ of $\gf[\sigma]$. Let $\modelA$ be a $\sigma$-model and let $(a_1,\dots,a_k)$ be a tuple of not necessarily distinct elements. Then whether or not
    \[\modelA \models \psi(a_1,\dots,a_k)\]
    holds depends only on the table of $(c_1,\dots,c_\ell)$, where $(c_1,\dots,c_\ell)$ is an arbitrary enumeration of the set of distinct elements of $(a_1,\dots,a_k)$.
\end{remark}

The definition of uniformity might seem bit complicated, but the following examples should clarify the intuition behind it.

\begin{example}
    Let $\sigma = \{S,R,P\}$, where $S$ is a ternary relation symbol, $R$ is a binary relation symbol and $P$ is a unary relation symbol. The formula
    \[\exists x \exists y (P(x) \land R(x,y) \land S(x,y,y) \land R(y,x) \land P(y)))\]
    is both uniform and one-dimensional. On the other hand the formula
    \[\exists x \exists y (\exists z (S(x,y,z) \land P(z)) \land R(x,y) \land S(x,y,x))\]
    is uniform but not one-dimensional. Finally, the formula
    \[\exists x \exists y \exists w (R(x,y) \land \exists z S(x,w,z))\]
    is neither one-dimensional nor uniform.
\end{example}

\begin{example}
    The standard translation of polyadic modal logic into $\mathcal{FO}$ results in formulas of the form
    \[\exists x_1 \dots \exists x_k (R(x_0,x_1,\dots,x_k) \land \bigwedge_{1\leq \ell \leq k} \psi_\ell(x_\ell))\]
    which are uniform and one-dimensional \cite{blackburn_rijke_venema_2001}.
\end{example}

We will use $\ugf$ to denote the set of formulas of $\gf$ which are uniform and $\guf1$ to denote the set of formulas of $\gf$ which are both uniform and one-dimensional. Throughout this paper we will use $\varphi(x_1,\dots,x_n)$, where all the variables in the tuple $(x_1,\dots,x_n)$ are distinct, to denote a formula of either $\guf1$ or $\ugf$ such that either $\{x_1,\dots,x_n\}$ is precisely the set of free variables of $\varphi$ or $\varphi$ has at most one free variable which belongs to $\{x_1,\dots,x_n\}$ or $\varphi$ is of the form $x_i = x_j$, where $1\leq i,j\leq n$.

\subsection{Interpolation}

We start by recalling the definition of the Craig interpolation property.

\begin{definition}
    Given a logic $\mathcal{L}$, we say that $\mathcal{L}$ has the Craig interpolation property (CIP), if for every $\varphi \in \mathcal{L}[\sigma]$ and $\psi \in \mathcal{L}[\tau]$ we have that $\varphi \models \psi$ implies that there exists an interpolant $\chi \in \mathcal{L}[\sigma \cap \tau]$ for this entailment, i.e., a sentence for which $\varphi \models \chi$ and $\chi \models \psi$ hold.
\end{definition}

It is well-known that the full $\gf$ fails to have CIP. The known examples of sentences which demonstrate this can be used to make the following observation.

\begin{proposition}
    The one-dimensional $\gf$ does not have CIP.
\end{proposition}
\begin{proof}
    Consider the following sentences, which are simple variants of the formulas used in \cite{Jung2021LivingWB}.
    \[\varphi := \exists x \exists y \exists z (G(x,y,z) \land R(x,y) \land R(y,z) \land R(z,x))\]
    \[\psi := \forall x \forall y (R(x,y) \to (A(x) \leftrightarrow \neg A(y)))\]
    Notice that both of these sentence are one-dimensional. Now one can show, using essentially the same argument as the one used in Example 1 in \cite{Jung2021LivingWB}, that there is no interpolant for the implication $\varphi\models \neg \psi$.
\end{proof}

We remark that, in the context of fragments of $\fo$, CIP is usually defined for \emph{formulas} instead of sentences (as we have defined it). We could have also formulated it for formulas, but we decided to work with sentences for simplicity.

\section{Bisimulation for $\guf1$}\label{BisimulationSection}

Given two models $\modelA$ and $\modelB$, and tuples $\overline{c}\in A^n$ and $\overline{d}\in B^n$ we will use $(\modelA, \overline{c}) \equiv_{\sigma} (\modelB, \overline{d})$ to denote the fact that for every $\varphi(x_1,\dots,x_n) \in \guf1$ we have that
\[\modelA \models \varphi(c_1,\dots,c_n) \iff \modelB \models \varphi(d_1,\dots,d_n).\]
The purpose of this section is to define a corresponding notion of bisimulation for $\guf1$ which captures the above equivalence relation. We will start by defining a suitable notion of partial isomorphism.

\begin{definition}\label{partialisomorphism}
    Let $\modelA$ and $\modelB$ be models, and let $X:=\{a_1,\dots,a_n\}\subseteq A$ and $Y\subseteq B$. A bijection $p:X \to Y$, is called a \emph{uniform partial $\sigma$-isomorphism} between $\modelA$ and $\modelB$, if \[\tp_\mathfrak{A}^\sigma[a_1,\dots,a_n] = \tp_\mathfrak{B}^\sigma[p(a_1),\dots,p(a_n)].\]
\end{definition}

Quantification in $\gf$ over a model $\modelA$ is restricted to \emph{live} subsets of $\modelA$, i.e., subsets of $\modelA$ which are either singletons or are \emph{contained} in a single tuple $\overline{a} \in R^{\modelA}$, for some $R\in \sigma$. In the case of $\guf1$ we will need the following modified version of the notion of live set, which takes into account the requirement that our formulas are uniform.

\begin{definition}\label{liveset}
    Let $\modelA$ be a model and let $X\subseteq A$. We say that $X$ is $\sigma$-live, if either $|X|\leq 1$ or there exists $R\in \sigma$ and $(a_1,\dots,a_n) \in R^{\modelA}$ so that $X = \{a_1,\dots,a_n\}$.
\end{definition}

We are now ready to define the notion of bisimulation for $\guf1$.

\begin{definition}\label{guardedbisimulation}
    Let $\bisim$ be a non-empty set of uniform partial $\sigma$-isomorphism between two structures $\modelA$ and $\modelB$. Let $\overline{c}\in A^n$ and $\overline{d}\in B^n$ be tuples. We say that $\bisim$ is a uniform guarded $\sigma$-bisimulation between $(\modelA, \overline{c})$ and $(\modelB, \overline{d})$, if for every $p:X \to Y \in \mathcal{Z}$ the following conditions hold:
    \begin{enumerate}
        \item [(cover)] There exists $h\in \bisim$ with $\overline{c} = \dom(h)$ so that $h(\overline{c}) = \overline{d}$.
        \item [(forth)] For any $a\in X$ and a $\sigma$-live set $X'\subseteq A$, with $a\in X'$, there exists $q:X' \to Y' \in \bisim$ so that 
        \[p(a) = q(a).\]
        \item [(back)] For any $b\in Y$ and a $\sigma$-live set $Y' \subseteq B$, with $b\in Y'$, there exists $q:X' \to Y' \in \mathcal{Z}$ so that 
        \[p^{-1}(b) = q^{-1}(b).\]
    \end{enumerate}
    If there exists a guarded $\sigma$-bisimulation between $(\modelA, \overline{c})$ and $(\modelB, \overline{d})$, then we denote this by $(\modelA, \overline{a}) \sim_\sigma (\modelB, \overline{b})$.
\end{definition}

In what follows we will often refer to uniform guarded bisimulations as guarded bisimulations. The following two lemmas establish that our notation of bisimulation is correct, the first of which can proved in a standard manner by using induction.

\begin{lemma}
    Let $\modelA$ and $\modelB$ be models, and let $\overline{c} \in A^n$ and $\overline{d} \in B^n$ be tuples so that $(\modelA, \overline{c}) \sim_\sigma (\modelB, \overline{d})$. Then $(\mathfrak{A}, \overline{c}) \equiv_\sigma (\mathfrak{B}, \overline{d})$.
\end{lemma}

For the proof of the second lemma we need to recall the definition of $\omega$-saturated model. A \emph{elementary} $n$-\emph{type} over a vocabulary $\sigma$ is a consistent set of first-order formulas (not necessarily quantifier-free) with free variables in $\{x_1,\dots,x_n\}$. Given a $\sigma$-model $\modelA$, we say that it is $\omega$-\emph{saturated}, if for every tuple $\overline{a}\in A^n$ of elements of $A$ we have that each elementary $n$-type over the extended vocabulary $\sigma \cup \{a_1,\dots,a_n\}$, where each $a_i$ denotes a constant to be interpreted as the element $a_i$, which is finitely consistent with the $\mathcal{FO}$-theory of $(\modelA, \overline{a})$, is realized in $(\modelA, \overline{a})$. It is well-known that every $\sigma$-model, where $\sigma$ is finite and relational, has an $\omega$-saturated elementary extension \cite{bell2006models}.

\begin{lemma}\label{upgradelemma}
    Let $\modelA$ and $\modelB$ be two $\omega$-saturated models, and let $\overline{c} \in A^n$ and $\overline{d} \in B^n$ be tuples so that $(\modelA, \overline{c}) \equiv_\sigma (\modelB, \overline{d})$. Then $(\modelA, \overline{c}) \sim_\sigma (\modelB, \overline{d})$.
\end{lemma}
\begin{proof}
    Consider the following set
    \[\bisim := \{p:\overline{a} \to \overline{b} \mid (\modelA, \overline{a}) \equiv_\sigma (\modelB, \overline{b})\}.\]
    We claim that $\bisim$ is a guarded $\sigma$-bisimulation between $(\modelA, \overline{c})$ and $(\modelB,\overline{d})$. We first note that by assumption $\overline{c} \mapsto \overline{d} \in \bisim$, and hence $\bisim$ satisfies (cover). $\bisim$ also clearly consists of uniform partial $\sigma$-isomorphism between $\modelA$ and $\modelB$. What remains to be proved is that $\bisim$ also satisfies (forth) and (back). Since these two cases are analogous, we will concentrate on (forth).
    
    Let $p:\overline{a} \to \overline{b} \in \bisim$, $a \in X$ and $X' := \{c_1,\dots,c_m\} \subseteq A$ be a $\sigma$-live set so that $a \in X'$. For simplicity we will assume that $a = c_1$. Consider now the following elementary $m$-type
    \[\Sigma := \{\varphi(p(a),x_2,\dots,x_m) \in \guf1[\sigma \cup \{p(a)\}] \mid \modelA \models \varphi(a,c_2,\dots,c_m)\}.\]
    We claim that $\Sigma$ is realized in $(\modelB,p(a))$. Since $\modelB$ is $\omega$-saturated, it suffices to show that each finite subset of $\Sigma$ is realized in $(\modelB,p(a))$. Let $\psi_1(p(a),x_2,\dots,x_m),\dots,\psi_r(p(a),x_2,\dots,x_m) \in \Sigma$. Since $X'$ is $\sigma$-live, there exists an atomic formula $\alpha(x_1,\dots,x_m)$ over $\sigma$ with the property that
    \[\modelA \models \exists x_2 \dots \exists x_m (\alpha(a,x_2,\dots,x_m) \land \bigwedge_{1\leq i\leq r} \psi_i(a,x_2,\dots,x_m)).\]
    Note that Definition \ref{liveset} guarantees that this is indeed a formula of $\guf1[\sigma]$. Since $(\modelA, \overline{a}) \equiv_\sigma (\modelB, \overline{b})$, we know that 
    \[\modelB \models \exists x_2 \dots \exists x_m (\alpha(p(a),x_2,\dots,x_m) \land \bigwedge_{1\leq i\leq r} \psi_i(p(a),x_2,\dots,x_m)).\]
    Thus $\{\psi_1(p(a),x_2,\dots,x_m),\dots,\psi_r(p(a),x_2,\dots,x_m)\}$ is satisfiable in $(\modelB,p(a))$, and hence $\Sigma$ is satisfiable in $(\modelB,p(a))$, say by the tuple $(p(a),d_2,\dots,d_m)$. Now $\overline{c} \mapsto \overline{d} \in \bisim$ is the mapping we were after. 
\end{proof}

\begin{remark}
    Using the two previous lemmas one prove in a standard manner that $\guf1$ is the maximal fragment of $\mathcal{FO}$ which is invariant under uniform guarded bisimulation, see for example \cite{Brny2018SOMEMT}.
\end{remark}

\section{$\guf1$ has CIP}\label{CipProofSection}

In this section we will prove that $\guf1$ has CIP. We will start with the following lemma.

\begin{lemma}\label{standard-compactness}
    Let $\sigma$ and $\tau$ be signatures, and let $\varphi \in \guf1[\sigma]$ and $\psi \in \guf1[\tau]$. Suppose that there is no $\chi \in \guf1[\sigma \cap \tau]$ with the property that $\varphi \models \chi$ and $\chi \models \psi$. Then there is a $\sigma$-model $\modelA$ and a $\tau$-model $\modelB$ with the property that $\modelA \models \varphi$, $\modelB \not \models \psi$ and $\modelA\equiv_{\sigma \cap \tau} \modelB$.
\end{lemma}
\begin{proof}
    Essentially the same argument as the one used in the proof of Theorem 4.1 in \cite{Brny2018SOMEMT} gives the result.
\end{proof}

To give a high level overview of the rest of the proof, suppose that the assumption of Lemma \ref{standard-compactness} holds for sentences $\varphi$ and $\psi$, which implies in particular that there are models $\modelA$ and $\modelB$ so that $\modelA \sim_{\sigma \cap \tau} \modelB$. Now, what we want to prove is that $\varphi \land \neg \psi$ is satisfiable. To do this, we will follow a standard approach in modal logic \cite{Brny2018SOMEMT,Hoogland2002InterpolationAD} by constructing an \emph{amalgam} $\amalgam$ which has the property that $\amalgam \sim_{\sigma} \modelA$ and $\amalgam \sim_{\tau} \modelB$. In particular, it will be a model of $\varphi \land \neg \psi$, since $\modelA\models \varphi$ and $\modelB \models \neg \psi$.

Suppose now that $\modelA \sim_{\sigma \cap \tau} \modelB$ and let $\bisim$ be a guarded $(\sigma \cap \tau)$-bisimulation which witnesses it. Given a pair $(\overline{a},\overline{b})$ we will use $(\overline{a},\overline{b})\in \bisim$ to denote the fact that there exists $p\in \bisim$ with the property that $\overline{a} = \dom(p)$ and $p(\overline{a}) = \overline{b}$. In other words the relation $a_i \mapsto b_i$ induces a uniform partial $(\sigma \cap \tau)$-isomorphism which belongs to $\bisim$.

Before describing the construction of $\amalgam$, we need to introduce some additional notation. Given two tuples $\overline{a}$ and $\overline{b}$ of the same length, we will let $(\overline{a} \otimes \overline{b})$ denote the following tuple:
\[((a_1,b_1),\dots,(a_n,b_n))\]
Given $(\overline{a} \otimes \overline{b})$, we say that it is \emph{left-good}, if for every $1\leq i < j\leq n$ we have that if $a_i = a_j$, then $b_i = b_j$. Similarly we say that $(\overline{a} \otimes \overline{b})$ is \emph{right-good}, if for every $1\leq i< j\leq n$ we have that if $b_i = b_j$, then $a_i = a_j$. Finally we say that $(\overline{a} \otimes \overline{b})$ is \emph{good} if it is left-good and right-good. Note that if $(\overline{a} \otimes \overline{b})$ is of length $n$, $k\geq n$ and $\mu:[k] \to [n]$ is a surjection, then we have that if $(\overline{a} \otimes \overline{b})$ is left-good, then so is $(\overline{a} \otimes \overline{b})_\mu$. Analogous observation of course holds for right-good and good.

As the domain of the amalgam $\amalgam$ we will take the set $U = (A \times B) \cap \bisim$, while the interpretations of relation symbols will be defined as follows. First, for every $R\in \sigma \cap \tau$ we define that
\[(\overline{a} \otimes \overline{b}) \in R^{\amalgam} \text{ iff $\overline{a} \in R^{\modelA}$ and $(\overline{a},\overline{b})\in \bisim$}\]
Then, for every $R \in (\sigma \backslash \tau)$ we define that $(\overline{a} \otimes \overline{b}) \in R^{\amalgam}$ iff $\overline{a} \in R^{\modelA}$ and one of the following conditions holds:
\begin{itemize}
    \item $(\overline{a},\overline{b}) \in \bisim$.
    \item $(\overline{a} \otimes \overline{b})$ is left-good and $\overline{a}$ is not $(\sigma \cap \tau)$-live.
\end{itemize}
Similarly, for every $R\in (\tau \backslash \sigma)$ we define that $(\overline{a} \otimes \overline{b}) \in R^{\amalgam}$ iff $\overline{b} \in R^{\modelB}$ and one of the following conditions holds:
\begin{itemize}
    \item $(\overline{a},\overline{b}) \in \bisim$.
    \item $(\overline{a} \otimes \overline{b})$ is right-good and $\overline{b}$ is not $(\sigma \cap \tau)$-live.
\end{itemize}
This concludes the construction of $\amalgam$. This construction is similar to the one given in \cite{Hoogland2002InterpolationAD} with the exception that we require tuples that are not $(\sigma \cap \tau)$-live to be either right-good or left-good.

We now define
\[\bisim_1 := \{(\overline{a} \otimes \overline{b}) \mapsto \overline{a} \mid \text{$(\overline{a}\otimes \overline{b})$ is $\sigma$-live in $\amalgam$.}\}\]
and
\[\bisim_2 := \{(\overline{a} \otimes \overline{b}) \mapsto \overline{b} \mid \text{$(\overline{a}\otimes \overline{b})$ is $\tau$-live in $\amalgam$.}\}\]
Note that if $(\overline{a} \otimes \overline{b})$ is $\sigma$-live, then by construction it is also left-good (and an analogous observation obviously holds for $\tau$-live tuples in $\amalgam$).

\begin{lemma}\label{lemmaPartialIsomorphism}
    $\bisim_1$ consists of uniform partial $\sigma$-isomorphism between $\amalgam$ and $\modelA$, and $\bisim_2$ consists of uniform partial $\tau$-isomorphism between $\amalgam$ and $\modelB$.
\end{lemma}
\begin{proof}
    We will only consider the case of $\bisim_1$, since the case of $\bisim_2$ is analogous. Let $(\overline{a} \otimes \overline{b}) \mapsto \overline{a} \in \bisim_1$, where the length of $(\overline{a} \otimes \overline{b})$ is $n$. We will check separately check that this mapping preserves $1$-types and $n$-ary atomic formulas. 
    
    Let $1\leq i\leq n$ and suppose that
    \[((a_i,b_i),\dots,(a_i,b_i)) \in R^{\amalgam},\]
    where $R\in \sigma$. By construction we know that $(a_i,\dots,a_i) \in R^{\amalgam}$. Suppose then that \[(a_i,\dots,a_i) \in R^{\modelA}.\]
    Since by definition of $U$ we have that $(a_i,b_i) \in \bisim$, we can conclude that $((a_i,b_i),\dots,(a_i,b_i)) \in R^{\amalgam}$. Thus $(a_i,b_i)$ and $a_i$ have the same $1$-types over $\sigma$.
    
    We will then verify that the mapping preserves $n$-ary atomic formulas. Let $R\in \sigma$ be a $k$-ary relation, where $k\geq n$, and let $\mu:[k] \to [n]$ be a surjection. We need to show that $(\overline{a} \otimes \overline{b})_\mu \in R^{\amalgam}$ iff $\overline{a}_{\mu} \in R^{\modelA}$. Again, the left to right direction follows immediately from the definition of $\amalgam$, so we will concentrate on the direction from right to left. First we note that if $\overline{a}$ is not $(\sigma \cap \tau)$-live, then we are done, since then also $\overline{a}_\mu$ is not $(\sigma \cap \tau)$-live.
    
    Thus we can assume that $\overline{a}$ is $(\sigma \cap \tau)$-live. Now, due to the definition of $\bisim_1$, we know that $(\overline{a} \otimes \overline{b})$ is $\sigma$-live in $\amalgam$. Hence, by definition of $\amalgam$, and the fact that $\overline{a}$ is $(\sigma \cap \tau)$-live, we know that $(\overline{a},\overline{b})\in \bisim$, which is the same as $(\overline{a}_{\mu},\overline{b}_{\mu}) \in \bisim$. Now we can deduce, due to the definition of $\amalgam$, that $(\overline{a}\otimes \overline{b})_\mu \in R^{\amalgam}$. This, together with the fact that $(\overline{a}\otimes \overline{b})\mapsto \overline{a}$ preserves $1$-types over $\sigma$, allows us to conclude that $\tp_{\amalgam}^{\sigma}[\overline{a} \otimes \overline{b}] = \tp_{\modelA}^{\sigma}[\overline{a}]$.
\end{proof}

\begin{lemma}\label{lemmaBackForth}
    $\bisim_1$ is a guarded $\sigma$-bisimulation between $\amalgam$ and $\modelA$, and $\bisim_2$ is a guarded $\tau$-bisimulation between $\amalgam$ and $\modelB$.
\end{lemma}
\begin{proof}
    Again, we will only consider the case of $\bisim_1$, since the case of $\bisim_2$ is analogous. Due to Lemma \ref{lemmaPartialIsomorphism} we just need to verify (back) and (forth) conditions. Let $(\overline{a} \otimes \overline{b}) \mapsto \overline{a} \in \bisim_1$, where the length of $\overline{a}$ and $\overline{b}$ is $n$.
    \begin{itemize}
        \item [(forth)] Let $(a_i,b_i) \in X$ and let $X' \subseteq U$ be a $\sigma$-live set so that $(a_i,b_i) \in X'$. Since $X'$ is $\sigma$-live, we know that it is of the form $\{(c_1,d_1),\dots,(c_m,d_m)\}$, with $(\overline{c} \otimes \overline{d})$ being left-good. Now $(\overline{c} \otimes \overline{d}) \mapsto \overline{c} \in \bisim_1$ is the required mapping.
        
        \item [(back)] Let $a_i \in Y$ and let $Y' \subseteq A$ be a $\sigma$-live set so that $a_i \in Y$. For concreteness, suppose that $Y' = \{c_1,\dots,c_m\}$. Consider first the case that $Y'$ is not $(\sigma \cap \tau)$-live in $\modelA$. For every $2\leq i\leq m$ we will pick an element $d_i$ so that $(c_i,d_i) \in \bisim$. Note that such elements exists since each singleton is a live element. By construction $\{(c_1,d_1),\dots,(c_m,d_m)\}$ is $\sigma$-live, and hence $(\overline{c} \otimes \overline{d}) \mapsto \overline{c} \in \bisim_1$ is the required mapping we were after.
        
        Suppose then that $Y'$ is $(\sigma \cap \tau)$-live in $\modelA$. Since $(a_i,b_i) \in U$, we know that $(a_i,b_i) \in \bisim$. Since $\bisim$ is a guarded $(\sigma \cap \tau)$-bisimulation, there exists a set $\{d_1,\dots,d_m\} \subseteq B$ so that $(\overline{c},\overline{d}) \in \bisim$ and $(a_i,b_i) \in (\overline{c} \otimes \overline{d})$. In particular $(\overline{c} \otimes \overline{d})$ is $\sigma$-live in $\amalgam$, and hence $(\overline{c} \otimes \overline{d}) \mapsto \overline{d} \in \bisim_1$, which is the mapping we were after.
    \end{itemize}
\end{proof}

\begin{theorem}
    $\guf1$ has Craig interpolation property.
\end{theorem}
\begin{proof}
    Let $\varphi \in \guf1[\sigma]$ and $\psi \in \guf1[\tau]$ be sentences so that $\varphi \models \psi$, but there is no interpolant for this entailment. By lemma \ref{standard-compactness} there exists a $\sigma$-model $\modelA$ and a $\tau$-model $\modelB$ such that $\modelA \models \varphi$, $\modelB \not\models \psi$ and $\modelA \equiv_{\sigma \cap \tau} \modelB$. Take $\omega$-saturated elementary extensions $\hat{\modelA}$ and $\hat{\modelB}$ of $\modelA$ and $\modelB$. Since $\hat{\modelA} \equiv_{\sigma \cap \tau} \hat{\modelB}$, by lemma \ref{upgradelemma} we have that $\hat{\modelA} \sim_{\sigma \cap \tau} \hat{\modelB}$. Using the construction presented in this section there exists a $(\sigma \cup \tau)$-model $\amalgam$ with the property that $\amalgam \sim_{\sigma} \hat{\modelA}$ and $\amalgam \sim_{\tau} \hat{\modelB}$. Thus $\amalgam \models \varphi \land \neg \psi$, i.e. $\varphi \land \neg \psi$ is consistent, which is a contradiction with the assumption that $\varphi \models \psi$.
\end{proof}

\section{Complexity of uniform $\gf$}\label{ComplexityUniformGfSection}

In this section we will prove that the complexity of the satisfiability problem of uniform $\gf$ is in $\textsc{NExpTime}$. Since it was proved in \cite{Kieronski2019OnedimensionalGF} that the complexity of the satisfiability problem of $\guf1$ is $\textsc{NExpTime}$-hard, this upper bound is sharp.

\subsection{Scott-normal form}

As usual, we will start by arguing that we can restrict our attention to sentences which are in a certain normal form. The normal form that we will use here has a somewhat awkward form, but the proof of Lemma \ref{complexity_scottnormalform} should clarify why we chose to use it.

\begin{definition}
    Let $\varphi$ be a sentence of $\ugf$. We say that $\varphi$ is in normal form, if it has the following shape
    \[\bigwedge_{t\in T} \exists z \lambda_t(z) \land \bigwedge_{i \in I} \forall \overline{x} (\alpha_i(\overline{x}) \to \exists \overline{y} (\beta_i(\overline{x},\overline{y}) \land \psi_i(\overline{x},\overline{y})))\]
    \[\land \bigwedge_{j\in J} \forall \overline{x} (\kappa_j(\overline{x}) \to (\theta_j(\overline{x}) \to \forall \overline{y} (\gamma_j(\overline{x},\overline{y}) \to \psi_j(\overline{x},\overline{y})))),\]
    where $T,I,J$ are non-empty (finite) sets, $\lambda_t$, $\alpha_i$, $\beta_i$, $\kappa_j$ and $\gamma_j$ are atomic formulas and $\psi_i$, $\theta_j$ and $\psi_j$ are quantifier-free formulas.
\end{definition}

\begin{remark}
    In the definition of the normal form we do not require that the tuples $\overline{y}$ are necessarily non-empty, i.e., we allow formulas of the form $\forall \overline{x} (\alpha_i(\overline{x}) \to \psi_i(\overline{x}))$ in our normal forms. However, we do require that the tuples $\overline{x}$ are non-empty, and hence we do not allow formulas of the form $\exists \overline{y} (\beta_i(\overline{y}) \land \psi_i(\overline{y}))$, where the length of $\overline{y}$ is more than one.
\end{remark}

If $\varphi$ is a sentence of $\ugf$ in normal form, then we refer to its conjuncts of the form
\[\forall \overline{x} (\alpha_i(\overline{x}) \to \exists \overline{y} (\beta_i(\overline{x},\overline{y}) \land \psi_i(\overline{x},\overline{y})))\]
as the \emph{existential requirements} and we will use $\varphi_i^\exists$ to denote them. Given a model $\modelA$, an existential requirement $\varphi_i^\exists$ and $\overline{a} \in \alpha_i^{\modelA}$ we say that a tuple $\overline{c}$ is a \emph{witness for $\varphi_i^\exists$ and $\overline{a}$} if 
\[\modelA \models \beta_i(\overline{a},\overline{c}) \land \psi_i(\overline{a},\overline{c}).\]
Conjuncts of the form
\[\forall \overline{x} (\kappa_j(\overline{x}) \to (\theta_j(\overline{x}) \to \forall \overline{y} (\gamma_j(\overline{x},\overline{y}) \to \psi_j(\overline{x},\overline{y}))))\]
will be referred to as the \emph{universal requirements} and we will use $\varphi_j^\forall$ to denote them.

Using standard renaming techniques one can establish the following.

\begin{lemma}\label{complexity_scottnormalform}
    There is a polynomial nondeterministic procedure, taking as its input a sentence $\varphi \in \ugf[\sigma]$ and producing a sentence $\varphi'\in \ugf[\sigma']$ in normal form, where $\sigma' \supset \sigma$, such that
    \begin{enumerate}
        \item if $\modelA \models \varphi$ for some model $\modelA$, then there is a run of the procedure producing a normal form $\varphi'$ such that $\modelA'\models \varphi'$ for some expansion $\modelA'$ of $\modelA$ to the vocabulary $\sigma'$,
        \item if the procedure has a run producing $\varphi'$ and $\modelA'\models \varphi'$, for some $\modelA'$, then the $\sigma$-reduct $\modelA$ of $\modelA'$ satisfies $\varphi$. 
    \end{enumerate}
\end{lemma}
\begin{proof}
    We will essentially follow the proof of lemma 1 in \cite{Kieronski2019OnedimensionalGF}, with some small technical modifications. Let $\varphi \in \ugf[\sigma]$ be a sentence, which w.l.o.g contains only existential quantification. Let $\psi$ be the innermost formula of $\varphi$ which starts with a block of existential quantifiers. If $\psi$ is a sentence, we will nondeterministically either replace it with $\bot$ or $\top$ and add $\psi$ or $\neg \psi$ (depending on our guess) as a conjunct to the resulting formula. Suppose then that $\psi$ is a formula of the form
    \[\exists \overline{y} (\alpha(\overline{x},\overline{y}) \land \psi(\overline{x},\overline{y})).\]
    Since $\varphi$ was a sentence, $\psi$ occurs in a scope of another formula of the form
    \[\exists \overline{z} (\alpha'(\overline{x}) \land \psi'(\overline{x})),\]
    where $\overline{z} \subseteq \overline{x}$. Let $\alpha'$ be the guard of the innermost such formula. We will now replace $\varphi$ with the following formula
    \[\varphi[\psi(\overline{x})/R(\overline{x})] \land \forall \overline{x} (R(\overline{x}) \to \exists \overline{y} (\alpha(\overline{x},\overline{y}) \land \psi(\overline{x},\overline{y})))\] \[\land \forall \overline{x} (\alpha'(\overline{x}) \to (\neg R(\overline{x}) \to \forall \overline{y} (\alpha(\overline{x},\overline{y}) \to \neg \psi(\overline{x},\overline{y})))),\]
    where $\varphi[\psi(\overline{x})/R(\overline{x})]$ is the sentence obtained from $\varphi$ by replacing the previously mentioned subformula $\psi(\overline{x})$ with the atomic formula $R(\overline{x})$ which has a fresh relation symbol $R$. It is straightforward to verify that the resulting sentence is equi-satisfiable with $\varphi$.
    
    Now one can repeat the above procedure until one is left with a sentence of the form
    \[\bigwedge_{t\in T} \exists \overline{x} (\alpha_t(\overline{x}) \land \psi_t(\overline{x})) \land \bigwedge_{i\in I} \varphi_i^\exists \land \bigwedge_{j\in J} \varphi_j^\forall,\]
    where each $\varphi_i^\exists$ is an existential requirement, while each sentence $\varphi_j^\forall$ is an universal requirement. Now one can replace each conjunct $\exists x_1 \dots \exists x_n (\alpha(x_1,\dots,x_n) \land \psi_t(x_1 \dots x_n))$ with a sentence of the form
    \[\exists x \lambda_t(x) \land \forall x_1 (\lambda_t(x_1) \to \exists x_2 \dots \exists x_n (\alpha_t(x_1,\dots,x_n) \land \psi_t(x_1,\dots,x_n))),\]
    where $\lambda_t$ is a fresh unary relation symbol. The resulting sentence is clearly equi-satisfiable with the original sentence and furthermore it is in normal form.
\end{proof}

\subsection{Satisfiability witnesses}

A standard technique in proving that the complexity of the satisfiability problem of a given fragment of $\fo$ is in \textsc{NExpTime} is to show that each satisfiable sentence of this fragment has a finite model of size at most exponential with respect to the length of the sentence \cite{GKV97,ORDEREDFRAGMENTS,Kieronski2019OnedimensionalGF,Kieronski2014ComplexityAE}. However, in the case of $\ugf$ it seems to be easier to show that we can associate to each of its sentences $\varphi$ a different type of certificate, which is still at most exponential with respect to the length of the sentence, and which can be used to construct a (potentially infinite) model for $\varphi$.

\begin{definition}\label{witness_definition}
    Let $\varphi \in \ugf[\sigma]$ be a sentence in normal form, $P$ be a set of $1$-types over $\sigma$ and $\pi \in P$. A pair $(\modelA, c)$, where $c\in A$, is called a $(P,\pi)$-witness for $\varphi$, if it satisfies the following requirements.
    \begin{enumerate}
        \item For every $a\in A$ we have that $\tp_{\modelA}[a] \in P$.
        \item For every existential requirement $\varphi_i^\exists$ and for every tuple $\overline{a}$ which contains $c$ we have that if $\modelA \models \alpha_i(\overline{a})$, then there exists a witness for $\varphi_i$ and $\overline{a}$.
        \item For every universal requirement $\varphi_j^\forall$ and for every tuple $\overline{a}$ which contains $c$ we have that if $\modelA \models \kappa_j(\overline{a}) \land \theta_j(\overline{a})$, then for every tuple $\overline{b}$ we have that \[\modelA \models \gamma_j(\overline{a},\overline{b}) \to \psi_j(\overline{a},\overline{b}).\]
    \end{enumerate}
\end{definition}

Here the intuition is that a $(P,\pi)$-witness $(\modelA,c)$ is a local certificate; it certifies that we can provide witnesses for tuples which contain the element $c$. The main idea now is that if we have a $(P,\pi)$-witness for each $\pi \in P$, then we can use them to construct a proper model for $\varphi$.

\begin{definition}
    Let $\varphi \in \ugf[\sigma]$ be a sentence in normal form. A set of $1$-types $P$ over $\sigma$ is a witness for $\varphi$, if it satisfies the following two requirements.
    \begin{enumerate}
        \item For every conjunct $\exists z \lambda_t(z)$ there exists $\pi \in P$ so that $\lambda_t(x) \in \pi$.
        \item For every $\pi \in P$ there exists a $(P,\pi)$-witness for $P$.
    \end{enumerate}
\end{definition}

The following lemmas prove that an existence of a witness for $\varphi$ is equivalent with the satisfiability of $\varphi$.

\begin{lemma}\label{lemma-witness-soundness}
    Let $\varphi\in \ugf$ be a sentence in normal form. If $\varphi$ is satisfiable, then there exists a witness for it.
\end{lemma}
\begin{proof}
    Suppose that $\modelA \models \varphi$. As the set of $1$-types $P$ we can take the set
    \[\{\tp_{\modelA}[a] \mid a\in A\}.\]
    Clearly for every conjunct $\exists z \lambda_t(z)$ there exists a suitable $1$-type in $P$. Towards verifying the second requirement let $\pi \in P$ and let $c\in A$ be an element which realizes $\pi$. Then $(\modelA,c)$ is clearly a $(P,\pi)$-witness for $\varphi$.
\end{proof}

\begin{lemma}\label{lemma-witness-completness}
    Let $\varphi \in \ugf$ be a sentence in normal form. If there exists a witness for $\varphi$, then it is satisfiable.
\end{lemma}
\begin{proof}
    For simplicity we will assume that $\varphi$ contains exactly one conjunct of the form $\exists z \lambda_t(z)$. Let $P$ be a witness for $\varphi$. Thus for every $\pi \in P$ there exists a pair $(\modelA^\pi,c)$ which is a $(P,\pi)$-witness for $\varphi$. Our goal is to use these witnesses to construct a sequence of models
    \[\modelA_1 \leq \modelA_2 \leq \modelA_3 \leq \dots\]
    so that their union is a model of $\varphi$.
    
    Let $\pi \in P$ be a $1$-type so that $\pi \models \lambda_t$. As the model $\modelA_1$ we will take the model which contains a single element with $1$-type $\pi$. Suppose then that we have defined $\modelA_n$ in such a way that each $1$-type realized in $\modelA_n$ belongs to $P$. To define the model $\modelA_{n+1}$ we will proceed as follows. Given $a\in \modelA_n$, we will use $W_a$ to denote the set $A^{\pi} - \{c\}$, where $A^{\pi}$ refers to the domain of the model in the $(P,\pi)$-witness $(\modelA^\pi,c)$ of $\pi := \tp_{\modelA_n}[a]$. Without loss of generality we will assume that the sets $W_a$ are pairwise disjoint. Now we will define $\modelA_{n+1}$ as follows.
    
    \begin{itemize}
        \item The domain of the model is
        \[A_n \cup \bigcup_{a\in A_n} W_a^*\]
        \item $\modelA_{n+1} \upharpoonright A_n$ is defined to be isomorphic with $A_n$.
        \item For each $a\in A_n$ and for each $\{c_1,\dots,c_m\}\subseteq W_a$, we define that
        \[\tp_{\modelA_{n+1}}[a,c_1,\dots,c_m] := \tp_{\modelA^{\pi}}[c,c_1,\dots,c_m],\]
        where $\pi$ is the $1$-type of $a$.
        \item For every tuple $(a_1,\dots,a_m)$ and a $m$-ary relation $R$ for which we have not yet defined whether $(a_1,\dots,a_m)$ belongs to $R^{\modelA_{n+1}}$, we will simply define that it does not belong to it.
    \end{itemize}
    
    \noindent The last step guarantees that if a tuple, which contains more than one element, is live in $\modelA_{n+1}$, then it was already alive in one of the models $\modelA^\pi$. It is straightforward to verify that the union of the models $(\modelA_n)_{n<\omega}$ is indeed a model of $\varphi$.
\end{proof}

\subsection{Complexity of $\ugf$}

Although the size of a witness for $\varphi$ is clearly only exponential with respect to $|\varphi|$, we do not yet have any upper bounds on the time it takes to verify that it really is a witness for $\varphi$. The following lemma gives us such a bound.

\begin{lemma}\label{lemma-upperbound-witness}
    Let $\varphi \in \ugf$ be a sentence in normal form and let $\sigma$ denote the vocabulary of $\varphi$. Let $P$ be a set of $1$-types over $\sigma$ and $\pi \in P$. If there exists a $(P,\pi)$-witness for $\varphi$, then there exists one in which the size of the model is at most $2^{|\varphi|^{O(1)}}$.
\end{lemma}
\begin{proof}
    Let $(\modelA,c)$ be a $(P,\pi)$-witness for $\varphi$ and let $m = \max \{ar(R) \mid R\in \sigma\}$. Note that $m\leq |\varphi|$. Our goal is to construct a sequence
    \[\modelB_1 \leq \dots \leq \modelB_m\]
    of models so that $(\modelB_m,c)$ is a $(P,\pi)$-witness for $\varphi$ and $|B_m|\leq 2^{|\varphi|^{O(1)}}$. As the model $\modelB_1$ we will take the model which contains a single element with $1$-type $\pi$; let $e$ denote this element.
    
    Before moving forward, we will introduce one auxiliary definition. Let $\overline{a} = (a_1,\dots,a_n)$ and $\overline{b} = (b_1,\dots,b_n)$ be tuples of elements from two models $\modelA$ and $\modelB$. Let $\{c_1,\dots,c_m\}$ denote the set of distinct elements in $\overline{a}$. We say that $\overline{a}$ and $\overline{b}$ are \emph{similar}, if the mapping $p:\overline{a} \to \overline{b}$, which was the mapping induced by the relation $a_i \mapsto b_i$, is a bijection and furthermore
    \[\tp_{\modelA}[c_1,\dots,c_n] = \tp_{\modelB}[p(c_1),\dots,p(c_n)].\]
    
    Suppose now that we have defined $\modelB_k$, where $k < m$, and in such a way that for each $\sigma$-live tuple $\overline{b}$ for which $\tp_{\modelB_k}[\overline{b}]$ has been defined, there exists a similar tuple $\overline{a}$ which consists of elements of $\modelA$. Given an existential requirement $\varphi_i^\exists$ of $\varphi$ and a tuple $\overline{b} \in \alpha_i^{\modelB_k}$, which contains the element $e$, we say that $\overline{b}$ is a \emph{$i$-defect} if there exists no witness for $\varphi_i^\exists$ and $\overline{b}$ in the model $\modelB_k$. By construction, for each $i$-defect $\overline{b}$ we can find a tuple $\overline{a}$ of elements of $\modelA$ so that $\overline{b}$ and $\overline{a}$ are similar. In particular $\overline{a} \in \alpha_i^{\modelA}$, and hence there exists a witness $\overline{c}$ for $\varphi_i^\exists$ and $\overline{a}$ in $\modelA$; let $W_{\overline{b},i}$ denote the set of elements in $\overline{c}$ which were not contained in $\overline{a}$. Without loss of generality we will assume that the sets $W_{\overline{b},i}$ are pairwise disjoint. Now we will define $\modelB_{k+1}$ as follows.
    
    \begin{itemize}
        \item The domain of the model is
        \[B_k \cup \bigcup_{i\in I} \bigcup_{\overline{b} \text{ an $i$-defect}} W_{\overline{b},i}\]
        \item $\modelB_{k+1} \upharpoonright B_k$ is defined to be isomorphic with $\modelB_k$.
        \item For each $i$-defect $\overline{b}$ and a set $W_{\overline{b},i} = \{c_1,\dots,c_n\}$ we define that 
        \[\tp_{\modelB_{k+1}}[d_1,\dots,d_r,c_1,\dots,c_n] = \tp_{\modelA}[p(d_1),\dots,p(d_r),c_1,\dots,c_n],\]
        where $(d_1,\dots,d_r)$ enumerates all the elements occurring in $\overline{b}$ and $p:\overline{b} \to \overline{a}$.
        \item For every tuple $(b_1,\dots,b_n)$ and a $n$-ary relation $R$ for which we have not yet defined whether $(b_1,\dots,b_n)$ belongs to $R^{\modelB_{k+1}}$, we will simply define that it does not belong to it.
    \end{itemize}
    
    This completes the construction of the models $\modelB_1,\dots,\modelB_m$. To bound the size of $\modelB_m$, we first note that $|B_{k+1}|\leq |B_k| + |\varphi||D_k|$, where $D_k$ denotes the number of defects in $\modelB_k$. By construction, for every defect $(d_1,\dots,d_r)$ of $\modelB_k$ the set $\{d_1,\dots,d_r\}$ is a $\sigma$-live set which is not contained in $\modelB_{\ell}$, for any $\ell < k$. If $k = 1$, then the number of such $\sigma$-live sets is one, and if $k > 1$, then the number of such $\sigma$-live sets is $D_{k-1}$. Since each $\sigma$-live set is of size at most $|\varphi|$, there are at most $|\varphi||\varphi|^{|\varphi|}D_{k-1} \leq 2^{|\varphi|^{O(1)}}D_{k-1}$ defects in $\modelB_k$, i.e., $D_k \leq 2^{|\varphi|^{O(1)}}D_{k-1}$. Thus $D_k\leq 2^{|\varphi|^{O(1)}}$, for any $k < m$, and hence $|B_m|\leq 2^{|\varphi|^{O(1)}}$.
    
    Thus what remains to be proven is that  $(\modelB_m,e)$ is a $(P,\pi)$-witness for $\varphi$. Here the only non-trivial requirement that we need to verify is that $\modelB_m$ satisfies the second item in definition \ref{witness_definition}. So, let $\varphi_i^\exists$ be an existential requirement and let $\overline{b} = (b_1,\dots,b_n) \in \alpha_i^{\modelB_m}$ be a tuple which contains $e$. We can clearly assume that $n < m$. It suffices to show that $\overline{b}$ is contained in $\modelB_k$, for some $k < m$, since then by construction we know that it has a witness in $\modelB_m$.
    
    Aiming for a contradiction, suppose that $\overline{b}$ is contained in $\modelB_m$, but it is not contained in $\modelB_k$ for any $k < m$. By construction we know that, since $\overline{b}$ is $\sigma$-live, we assigned a table to some tuple $(b_1',\dots,b_r')$, where $(b_1',\dots,b_r')$ enumerates the set of distinct elements of $(b_1,\dots,b_n)$. Again, by construction we know that we assigned a table to the tuple $(b_1',\dots,b_r')$, because we wanted to provide a witness for some tuple $(d_1,\dots,d_s)$, which contains $e$ and for which $\{d_1,\dots,d_s\}$ is a \emph{strict} subset of $\{b_1',\dots,b_r'\}$.\footnote{If it were not, there would have been no need to provide a witness for it.}
    
    Now observe that $(d_1,\dots,d_s)$ is a $\sigma$-live tuple containing $e$, which is contained in $\modelB_{m-1}$ but is not contained in $\modelB_k$ for any $k < m - 1$. Indeed, if it were contained in $\modelB_k$, for some $k < m - 1$, then by construction we would have provided a witness for it in the model $\modelB_{k+1}$, i.e., $(b_1,\dots,b_n)$ would have been contained in $\modelB_{k+1}$. But now we are in a position which is the same as the one that we started in; in particular, we can repeat the above argument. After repeating the argument (at least) $(n-1)$-times we would end up with the conclusion that $e$ is contained in some $\modelB_k$, where $k > 1$, but it is not contained in $\modelB_1$, which would be an obvious contradiction.
\end{proof}

Now we can prove the main theorem of this section.

\begin{theorem}\label{ThmUniformGfComplexity}
    The complexity of the satisfiability problem of $\ugf$ is \textsc{NExpTime}-complete.
\end{theorem}
\begin{proof}
    The lower bound follows from the proof of Theorem 3 in \cite{Kieronski2019OnedimensionalGF}. We will give an informal description of a non-deterministic procedure running in exponential time which determines whether a given sentence $\varphi \in \ugf$ is satisfiable. It starts by converting $\varphi$ into an equi-satisfiable sentence $\varphi' \in \ugf$ in normal form, after which it guesses a set of $1$-types $P$ over the vocabulary of $\varphi'$ and for each $\pi \in P$ a $(P,\pi)$-witness $(\modelA,c)$ for $\varphi$, where the size of $\modelA$ is at most $2^{|\varphi|^{O(1)}}$. Lemmas \ref{complexity_scottnormalform}, \ref{lemma-witness-soundness}, \ref{lemma-witness-completness} and \ref{lemma-upperbound-witness} guarantee that this procedure is correct. Since $|P|\leq 2^{|\varphi|}$, the algorithm runs in exponential time with respect to $|\varphi|$.
\end{proof}

\section{Conclusions}

In this paper we have proved two results of quite distinct flavour on uniform guarded fragments. The first result was that although $\gf$ fails to have Craig interpolation, its one-dimensional uniform fragment does have it. The second result was that the complexity of the satisfiability problem of the uniform guarded fragment is \textsc{NExpTime}-complete. The results presented in this paper suggest several new research questions, but here we will mention just two of them.

The first question is whether or not the uniform $\gf$ has Craig interpolation property. While the correctness of the amalgam construction presented in section \ref{CipProofSection} rests on the assumption of one-dimensionality, we have not been able to show that uniform $\gf$ would not have Craig interpolation property. This has led the author to conjecture that the uniform $\gf$ does in fact have Craig interpolation property.

The second question is whether or not uniform $\gf$ has the exponential model property (note that if uniform $\gf$ would have an exponential model property, then one would obtain theorem \ref{ThmUniformGfComplexity} for free). As we saw in the proof of lemma \ref{lemma-upperbound-witness}, the requirement of uniformity essentially prevents uniform $\gf$ from enforcing long paths, and this seems to suggest that uniform $\gf$ can only enforce exponentially long paths (which it can enforce, since it contains standard modal logic with the global diamond). Because of this, the author conjectures that uniform $\gf$ has the exponential model property.

\section*{Acknowledgements}

The author wishes to thank Bartosz Bednarczyk for several helpful discussions on interpolation and fragments of first-order logic, and for suggesting the problem of determining the complexity of the satisfiability problem of uniform guarded fragment. The author also wishes to thank Antti Kuusisto for pointing out rather silly mistakes in the original definitions of uniformity and the uniform guarded bisimulation.

\bibliographystyle{splncs04}
\bibliography{fossacs}

\end{document}